\documentclass[a4paper]{article}

\usepackage[
	left=3.5cm,
	right=3.5cm,
	top=2cm,
	bottom=2.5cm,
	includeheadfoot]{geometry}

\usepackage[english]{babel}
\usepackage[T1]{fontenc}
\usepackage[utf8]{inputenc}
\usepackage{lmodern}

\newtheorem{theorem}{Theorem}

\newtheorem{lemma}[theorem]{Lemma}
\newtheorem{cor}[theorem]{Corollary}

\newtheorem{conj}{Conjecture}
\newtheorem{prop}[theorem]{Proposition}

\def\QED{\ensuremath{{\square}}}
\def\markatright#1{\leavevmode\unskip\nobreak\quad\hspace*{\fill}{#1}}
\newenvironment{proof}
{\begin{trivlist}\item[\hskip\labelsep{\bf Proof.}]}
	{\markatright{\QED}\end{trivlist}}

\usepackage{graphicx}
\usepackage{todonotes}
\usepackage{amsmath,amssymb}
\usepackage{mathtools}

\newtheorem{definition}[theorem]{Definition}
\newenvironment{refproof}
{\begin{trivlist}\item[]}
	{\markatright{\QED}\end{trivlist}}
\newcommand{\tk}[2]{\hat{E}_{#1}(#2)}%
\newcommand{\dk}[2]{\bar{E}_{#1}(#2)}%
\newcommand{\sk}[2]{E_{#1}(#2)}%
\newcommand{\di}[2]{\bar{I}_{#1}(#2)}%
\newcommand{\si}[2]{I_{#1}(#2)}%
\newcommand{\td}[2]{\hat{\Delta}_{#1}(#2)}
\newcommand{\dd}[2]{\bar{\Delta}_{#1}(#2)}
\newcommand{\sd}[2]{\Delta_{#1}(#2)}
\newcommand{\hd}[1]{\Delta_{cr}(#1)}
\newcommand{\mm}{\lfloor\frac{n}{2}\rfloor-2}%
\newcommand{\mmm}{\lfloor\frac{n}{2}\rfloor-3}%
\newcommand{\ssps}{semi-pair-shellable}%
\newcommand{\sspsy}{semi-pair-shellability}%
\newcommand{\sspsG}{Semi-Pair-Shellable}%
\newcommand{\sspsyG}{Semi-Pair-Shellability}%
\newcommand{\Sspsy}{Semi-pair-shellability}%
\newcommand{\psy}{pair-shellability}%
\title{The Crossing Number of \sspsG{} Drawings \\of Complete Graphs} 
\author{ 
	Petra Mutzel 
	\thanks{Department of Computer Science, TU Dortmund University, {\tt petra.mutzel@tu-dortmund.de}}
	\and 
	Lutz Oettershagen 
	\thanks{Department of Computer Science, TU Dortmund University, {\tt lutz.oettershagen@tu-dortmund.de}}
}
\index{Mutzel, Petra}
\index{Oettershagen, Lutz}
\begin{document}
\maketitle
\begin{abstract}
	The Harary-Hill Conjecture states that for $n\geq 3$ every drawing of $K_n$ has at least
	\begin{align*}
	H(n) \coloneqq \frac{1}{4}\Big\lfloor\frac{n}{2}\Big\rfloor\Big\lfloor\frac{n-1}{2}\Big\rfloor\Big\lfloor\frac{n-2}{2}\Big\rfloor\Big\lfloor\frac{n-3}{2}\Big\rfloor
	\end{align*} 
	crossings. 
	In general the problem remains unsolved, however there has been some success in proving the conjecture for restricted classes of drawings.
	The most recent and most general of these classes is seq-shellability~\cite{seqshellable}. 
	In this work, we improve these results and introduce the new class of \emph{\ssps{}} drawings.
	We prove the Harary-Hill Conjecture for this new class using novel results on $k$-edges. 
	So far, approaches for proving the Harary-Hill Conjecture for specific classes rely on a fixed reference face. 
	We successfully apply new techniques in order to loosen this restriction,
	which enables us to select different reference faces when considering subdrawings.
	Furthermore, we introduce the notion of \mbox{\emph{$k$-deviations}} as the difference between an optimal and the actual number of $k$-edges.
	Using $k$-deviations, we gain interesting insights into the essence of $k$-edges, and we further relax the necessity of fixed reference faces.
\end{abstract}
\section{Introduction}
The crossing number $cr(G)$ of a graph $G$ is the smallest number of edge crossings over all possible drawings of~$G$. 
In a drawing $D$ of $G=(V,E)$ every vertex $v \in V$ is represented by a point and every edge $uv\in E$ with $u,v \in V$ is represented by a simple curve connecting the corresponding points of $u$ and $v$.
We call an intersection point of the interior of two edges a crossing. 
The Harary-Hill Conjecture states the following.
\begin{conj}[Harary-Hill \cite{Guy1960}]
	Let $K_n$ be the complete graph with $n$ vertices, then
	\begin{align*}
	cr(K_n) = H(n) &\coloneqq \frac{1}{4}\Big\lfloor\frac{n}{2}\Big\rfloor\Big\lfloor\frac{n-1}{2}\Big\rfloor\Big\lfloor\frac{n-2}{2}\Big\rfloor\Big\lfloor\frac{n-3}{2}\Big\rfloor
	\end{align*} 
\end{conj}
There are construction methods for drawings of $K_n$ that lead to exactly $H(n)$ crossings,
for example the class of \emph{cylindrical} drawings first described by Hill~\cite{hararyhill1963}.
However, there is no proof for the lower bound of the conjecture for arbitrary drawings of $K_n$ with $n\geq 13$.
The cases for $n\leq 10$ have been shown by Guy \cite{Guy1960} and for $n=11$ by Pan and Richter \cite{DBLP:journals/jgt/PanR07}. 
Guy \cite{Guy1960} argues that $cr(K_{2n+1})\geq H(2n+1)$ implies $cr(K_{2(n+1)})\geq H(2(n+1))$, hence $cr(K_{12})\geq H(12)$.
McQuillan et al.~\cite{DBLP:journals/jct/McQuillanPR15} showed that $cr(K_{13})\geq 219$. 
{\'A}brego et al.~\cite{abrego2015all} improved the result to $cr(K_{13})\in\{223,225\}$. 

Beside these results for arbitrary drawings, there has been success in proving the Harary-Hill Conjecture for different classes of drawings. 
So far, the conjecture has been verified for 2-page-book \cite{Abrego:2012:CNK:2261250.2261310}, cylindrical \cite{DBLP:journals/dcg/AbregoAF0S14}, $x$-monotone \cite{DBLP:journals/dcg/BalkoFK15,ABREGO2013411}, $x$-bounded  \cite{DBLP:journals/dcg/AbregoAF0S14}, shellable \cite{DBLP:journals/dcg/AbregoAF0S14}, bishellable \cite{AbregoAFMMM0RV15} and recently seq-shellable drawings \cite{seqshellable}.

Seq-shellability is the broadest of the beforehand mentioned classes comprising the others.
%
Here, the proof of the Harary-Hill Conjecture makes use of the concept of \emph{$k$-edges}. 
Each edge $e\in E$ in a drawing is assigned a specific value between $0$ and $\lfloor\frac{n}{2}\rfloor-1$ w.r.t. a fixed reference face. 
The edge $e$ separates the remaining $n-2$ to vertices into two distinct sets, and 
is assigned the cardinality $k$ of the smaller of the two sets, i.e. is a $k$-edge
(see section \ref{sec:soa} for details).
We can express the number of crossings in a drawing in terms of the numbers of $k$-edges for each $k\in\{0,\ldots,\lfloor\frac{n}{2}\rfloor-2\}$.  
Therefore, having lower bounds on the (cumulated) number of $k$-edges implies a lower bound on the crossing number of a drawing.
After two cumulations, 
we obtain \emph{double cumulated $k$-edges}. 
However, the possibilities of their usage for further improvements to new classes of drawings seem to be limited.
\paragraph{Our contribution and outline}
In this work, we resolve the limitations of double cumulated $k$-edges by applying two new ideas. 
Firstly, instead of double cumulated $k$-edges we utilize \emph{triple cumulated} $k$-edges.
Balko et al. introduced these 
in \cite{DBLP:journals/dcg/BalkoFK15}.
Secondly, so far all classes, including seq-shellability, depend on a \emph{globally fixed} reference face.
We call a reference face globally fixed if we do not allow to select a different one when considering subdrawings, which constitutes a strong limitation in the proofs. 
In this work, we show that under certain conditions and/or assumptions, we are able to change the reference face \emph{locally} or even without restrictions.
Changing the reference face locally means, given a vertex $v$ incident to the initial reference face, we select a new reference face $F$, such that $F$ is also incident to~$v$.
Furthermore, we introduce a new class of drawings for which we verify the Harary-Hill Conjecture; we call drawings belonging to this class \emph{\ssps{}}.
There are \ssps{} drawings that are not seq-shellable.
But unlike seq-shellability, \sspsy{} does not comprise all previously found classes and only contains drawings with an odd number of vertices. 
However, every $(\lfloor\frac{n}{2}\rfloor-1)$-seq-shellable drawing with $n$ odd is \ssps{}.
Moreover, we introduce \emph{$k$-deviations} of a drawing $D$ of $K_n$.
They are the difference between the numbers of cumulated $k$-edges in $D$ and reference values corresponding to a drawing with exactly $H(n)$ crossings.
They allow us to further relax the necessity of a globally fixed reference face.

The outline of this paper is as follows. In Section \ref{sec:soa} we introduce the preliminaries, and in particular the necessary background on (cumulated) $k$-edges and their usage for verifying the Harary-Hill Conjecture.
In the following Section \ref{sec:triplecumu}, we present our novel results for triple cumulated $k$-edges, followed by the introduction of \ssps{} drawings in Section \ref{sec:pairseqshell}.
We verify the Harary-Hill Conjecture for this class, and discuss the distinctive differences to seq-shellability.
In Section \ref{sec:kdev} we use $k$-deviations to formulate conditions under which we are able to further loosen the need for a globally fixed reference face.
We conjecture these conditions to be true in all good drawings. Assuming our conjecture holds, we prove the Harary-Hill Conjecture for a broad class of drawings that also comprises seq- and \sspsy{}. 
Finally, in Section \ref{sec:conclusions} we draw our conclusions and give an outlook to further possible work.
\section{Preliminaries}\label{sec:soa}
A \emph{drawing} $D$ of a graph $G$ on the plane is an injection $\phi$ from the vertex set $V$ into the plane, and a mapping of the edge set $E$ into the set of simple curves, such that the curve corresponding to the edge $e = uv$ has endpoints $\phi(u)$ and $\phi(v)$, and contains no other vertices \cite{szekely2000successful}. 
We call an intersection point of the interior of two edges a crossing and a shared endpoint of two adjacent edges is not considered a crossing. 
The crossing number $cr(D)$ of a drawing $D$ equals the number of crossings in $D$ and the crossing number $cr(G)$ of a graph $G$ is the minimum crossing number over all its possible drawings.
We restrict our discussions to \emph{good} drawings of $K_n$, and call a drawing \emph{good} if $(1)$ any two of the curves have finitely many points in common, $(2)$ no two curves have a point in common in a tangential way, $(3)$ no three curves cross each other in the same point, $(4)$ any two edges cross at most once and $(5)$ no two adjacent edges cross. It is known that every drawing with a minimum number of crossings is good \cite{schaefer2013graph}.
In the discussion of a drawing $D$, we call the points also vertices, the curves edges and 
$V$ denotes the set of vertices (i.e. points), and $E$ denotes the edges (i.e curves) of $D$. 
If we subtract the drawing $D$ from the plane, a set of open discs remain. 
We call $\mathcal{F}(D) \coloneqq \mathbb{R}^2 \setminus D$ the set of \emph{faces} of the drawing $D$.
If we remove a vertex $v$ and all its incident edges from $D$, we get the subdrawing $D-v$. 
We denote with $f(v)$ the unique face in $D-v$ that contains all the faces that are incident to $v$ in $D$, and call $f(v)$ the \emph{superface} of $v$.
We might consider the drawing to be on the surface of the sphere $S^2$, which is equivalent to the drawing on the plane due to the homeomorphism between the plane and the sphere minus one point. 
Next, we introduce \emph{$k$-edges};
the origins of $k$-edges lie in computational geometry and problems over $n$-point set, especially problems on halving lines and $k$-set \cite{abrego2012k}.
An early definition in the geometric setting goes back to Erd\H{o}s et al \cite{erdos1973dissection}.
Given a set $P$ of $n$ points in general position in the plane,
the authors add a directed edge $e=(p_i,p_j)$ between the two distinct points $p_i$ and $p_j$, and consider the continuation as line that separates the plane into the left and right half plane. There is a (possibly empty) point set $P_L\subseteq P$ on the left side of $e$, i.e. left half plane. Erd\H{o}s et al. assign $k:= \min(|P_L|, |P\setminus P_L|)$ to $e$.
Later, the name $k$-edge emerged and Lov{\'a}sz et al.~\cite{lovasz2004convex} used $k$-edges for determining a lower bound on the crossing number of rectilinear graph drawings. Finally, 
{\'A}brego et al.~\cite{Abrego:2012:CNK:2261250.2261310} extended the concept of $k$-edges from rectilinear to topological graph drawings. 
Every edge in a good drawing $D$ of $K_n$ is a $k$-edge for a specific value of $k\in\{0,\ldots,\lfloor\frac{n}{2}\rfloor-1\}$. Let $D$ be on the surface of the sphere $S^2$, and $e=uv$ be an edge in $D$ and $F\in\mathcal{F}(D)$ be an arbitrary but fixed face; we call $F$ the \emph{reference face}.
Together with any vertex $w\in V\setminus\{u,v\}$, the edge $e$ forms a triangle $uvw$ and hence a closed curve that separates the surface of the sphere into two parts.
For an arbitrary but fixed orientation of $e$, one can distinguish between the left part and the right part of the separated surface. 
If $F$ lies in the left part of the surface, we say the triangle has orientation $+$ else it has orientation $-$. 
For $e$ there are $n-2$ possible triangles in total, of which $0\leq i\leq n-2$ triangles have orientation $+$ (or $-$) and $n-2-i$ triangles have orientation $-$ (or $+$ respectively).
We define the \emph{$k$-value} of $e$ to be the minimum of $i$ and $n-2-i$.
We say $e$ is an \emph{$i$-edge} with respect to the reference face $F$ precisely if its $k$-value equals $i$. See Figure \ref{fig:example} for an example.
{\'A}brego et al.~\cite{Abrego:2012:CNK:2261250.2261310} show that the crossing number of a drawing is expressible in terms of the number of $k$-edges for $0\leq k\leq \lfloor\frac{n}{2}-1\rfloor$ with respect to the reference face. 
The following definitions of the \emph{cumulated} numbers of $k$-edges are
used for determining lower bounds of the crossing number. The double cumulated number of $k$-edges has been defined in \cite{Abrego:2012:CNK:2261250.2261310}, and the triple cumulated number of $k$-edges has been introduced by Balko et al~\cite{DBLP:journals/dcg/BalkoFK15}.
\begin{definition}\emph{\cite{Abrego:2012:CNK:2261250.2261310,DBLP:journals/dcg/BalkoFK15}}
	Let $D$ be good drawing and $\sk{k}{D}$ be the number of $k$-edges in $D$ with respect to a reference face $F\in \mathcal{F}(D)$ and for each $k\in\{0,\ldots, \lfloor\frac{n}{2}\rfloor-1\}$.
	We denote
	\begin{align*}
		\dk{k}{D} \coloneqq \sum_{j=0}^{k} \sum_{i=0}^{j}\sk{i}{D}=\sum_{i=0}^{k}(k+1-i) \sk{i}{D}
	\end{align*}
	the \emph{double cumulated number of $k$-edges}, and 
	\begin{align*}
		\tk{k}{D} \coloneqq \sum_{i=0}^{k}\dk{i}{D} = \sum_{i=0}^{k}{k+2-i \choose 2} \sk{i}{D}
	\end{align*}
	the \emph{triple cumulated number of $k$-edges}. 	
\end{definition}
We also write \emph{double (triple) cumulated $k$-edges} or \emph{double (triple) cumulated $k$-value}
instead of double (triple) cumulated number of $k$-edges. We express the crossing number of a drawing using the triple cumulated $k$-edges.
\begin{theorem}\label{theorem:triplecumubound}\emph{\cite{DBLP:journals/dcg/BalkoFK15}} Let $D$ be a good drawing of $K_n$ and $m=\lfloor \frac{n}{2} \rfloor -2$. With respect to a reference face $F\in \mathcal{F}(D)$ we have for $n$ odd
	\begin{align*}
	cr(D) = 2\cdot \tk{m}{D} - \frac{1}{8}n(n-1)(n-3)
	\end{align*} 
	and $n$ even
	\begin{align*}
	cr(D) =& \tk{m}{D}+\tk{m-1}{D} - \frac{1}{8}n(n-1)(n-2) \textrm{.}
	\end{align*} 
\end{theorem}
Consequently, for $n$ odd $\tk{m}{D}$ and $n$ even $\tk{m}{D}+\tk{m-1}{D}$ are unique for all faces of $D$. 
Notice, this does not apply to the double cumulated case, i.e. $\dk{m}{D}$ or $\dk{m}{D}+\dk{m-1}{D}$, respectively.
Using the following lower bounds, we are able to verify the Harary-Hill Conjecture.
\begin{cor}\emph{\cite{DBLP:journals/dcg/BalkoFK15}~}\label{corollary:triple_cumu_lower_bound}
	Let $D$ be a good drawing of $K_n$. 
	If $n$ is odd and 
	\begin{align*}
	\tk{\frac{n-1}{2}-2}{D}\geq 3{\frac{n-1}{2}+2 \choose 4}			
	\end{align*}
	or $n$ is even and with respect to a face $F\in\mathcal{F}(D)$
	\begin{align*}
	\tk{\frac{n}{2}-2}{D} \geq 3{\frac{n}{2}+2 \choose 4}\text{~and~~} \tk{\frac{n}{2}-3}{D}\geq 3{\frac{n}{2}+1 \choose 4},	
	\end{align*}
	then $cr(D)\geq H(n)$.
\end{cor}
If a vertex touches the reference face, it is incident to a predetermined set of $k$-edges.
\begin{lemma}\label{lemma:vatf}\emph{\cite{Abrego:2012:CNK:2261250.2261310}}
	Let $D$ be a good drawing of $K_n$, $F\in \mathcal{F}(D)$ and $v \in V$ be a vertex incident to $F$.
	With respect to $F$, vertex $v$ is incident to two $i$-edges for $0\leq i \leq \lfloor \frac{n}{2}\rfloor -2$.
	Furthermore, if we label the edges incident to $v$ counter clockwise with $e_0,\ldots,e_{n-2}$ such that $e_0$ and $e_{n-2}$ are incident to the face $F$, then $e_i$ is a $k$-edge with $k=\min(i,n-2-i)$ for $0\leq i \leq n-2$.
\end{lemma}
The definition of \sspsy{} uses seq-shellability, which itself is based on simple sequences.
\begin{definition}[Simple sequence]\emph{\cite{seqshellable}}
	Let $D$ be a good drawing of $K_n$, $F\in\mathcal{F}(D)$ and $v\in V$ with $v$ incident to $F$. 
	Furthermore, let $S_v = (u_0,\ldots,u_k)$ with $u_i\in V\setminus\{v\}$ be a sequence of distinct vertices.
	If $u_0$ is incident to $F$ and vertex $u_i$ is incident to a face containing $F$ in the subdrawing $D-\{u_0,\ldots, u_{i-1}\}$ for all $1\leq i \leq k$, then we call $S_v$ simple sequence of $v$.
\end{definition}
\begin{definition}[Seq-Shellability]\emph{\cite{seqshellable}}
	Let $D$ be a good drawing of $K_n$. We call $D$ \emph{$k$-seq-shellable} for $k\geq 0$ if there exists a face $F\in \mathcal{F}(D)$ and a sequence of distinct vertices $a_0,\ldots,a_k$ such that $a_0$ is incident to $F$, and
	\emph{(1.)} for each $i\in \{1,\ldots,k\}$, vertex $a_i$ is incident to the face containing $F$ in drawing $D-\{a_0,\ldots,a_{i-1}\}$, and
	\emph{(2.)} for each $i\in \{0,\ldots,k\}$, vertex $a_i$ has a simple sequence $S_i=(u_0,\ldots,u_{k-i})$ with $u_j\in V\setminus \{a_0,\ldots,a_i\}$ for $0\leq j \leq k-i $ in drawing $D-\{a_0,\ldots,a_{i-1}\}$. 
\end{definition}
If a drawing $D$ of $K_n$ is $(\lfloor\frac{n}{2}\rfloor-2)$-seq-shellable, we omit the $(\lfloor\frac{n}{2}\rfloor-2)$ part and say $D$ is seq-shellable. 
The class of seq-shellable drawings contains all drawings that are $(\lfloor\frac{n}{2}\rfloor-2)$-seq-shellable.
\begin{figure}
	\centering
	\includegraphics[width=0.8\linewidth]{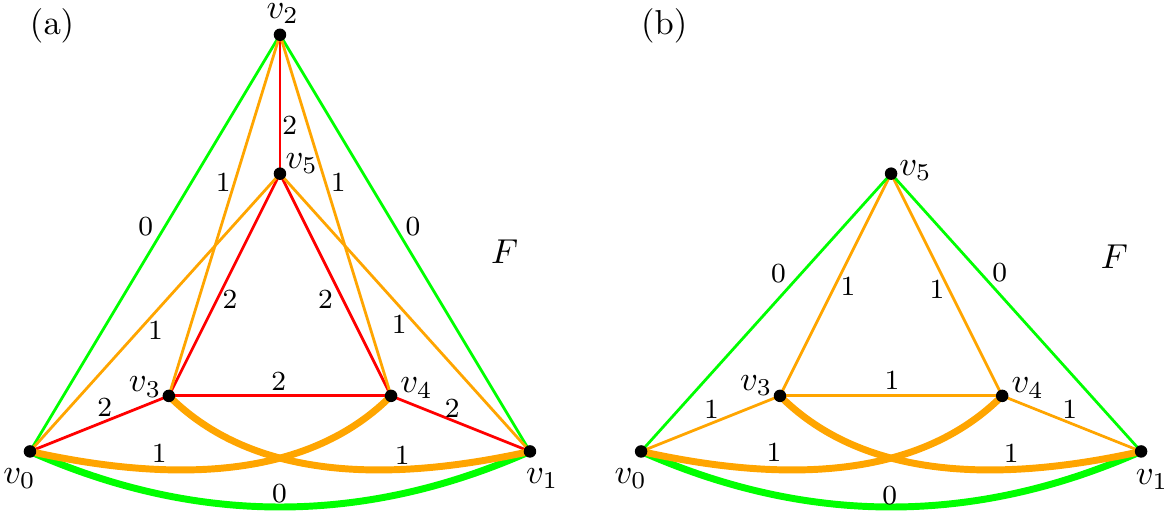}
	\caption{Example (a) shows a crossing optimal drawing $D$ of $K_6$ with the $k$-values at the edges. (b) shows the subdrawing $D-v_2$ and its $k$-values. The fat highlighted edges $v_0v_1$, $v_0v_4$ and $v_1v_3$ are invariant and keep their $k$-values. The reference face is the outer face $F$.}
	\label{fig:example}
\end{figure}
\section{Properties of Triple Cumulated $k$-Edges}
\label{sec:triplecumu}
In this section, we present new results for triple cumulated $k$-edges.
First, we introduce the triple cumulated value of edges incident to $v$.
Having a vertex $v$ incident to the reference face $F$, we know from Lemma \ref{lemma:vatf} that $v$ is incident to two $k$-edges for each $k\in \{0,\ldots, \lfloor\frac{n}{2}\rfloor-2\}$ and it follows that the triple cumulated number of $k$-edges incident to $v$ is 
$\tk{k}{D,v}=\sum_{i=0}^{k}{k+2-i \choose 2}\cdot 2 = 2{k+3 \choose 3}$.

Next, we introduce the \emph{double cumulated invariant} edges.
Consider removing a vertex $v \in V$ from a good drawing $D$ of $K_n$, resulting in the subdrawing $D-v$.  By deleting $v$ and its incident edges every remaining edge loses one triangle, i.e. for an edge $uw\in E$ there are only $(n-3)$ triangles $uwx$ with $x\in V\setminus\{u,v\}$ (instead of the $(n-2)$ triangles in drawing $D$). 
The $k$-value of any edge $e\in E$ is defined as the minimum number of $+$ or $-$ oriented triangles that contain $e$. If the lost triangle had the same orientation as the minority of triangles, the $k$-value of $e$ is reduced by one else it stays the same.
Therefore, every $k$-edge in $D$ with respect to $F\in \mathcal{F}(D)$ is either a $k$-edge or a $(k-1)$-edge in the subdrawing $D-v$ with respect to $F'\in \mathcal{F}(D-v)$ and $F\subseteq F'$.
We call an edge $e$ \emph{invariant} if $e$ has the same $k$-value with respect to $F$ in $D$ as for $F'$ in $D'$. 
See Figure \ref{fig:example} for an example.

For $0\leq k \leq \lfloor\frac{n}{2}\rfloor-1$ we denote the number of invariant $k$-edges between $D$ and $D'$ (with respect to $F$ and $F'$ respectively) with $I_k(D, D')$.
Furthermore, we define the \emph{double cumulated invariant $k$-value} as
\begin{align*}
\di{k}{D, D'} \coloneqq \sum_{j=0}^{k}\sum_{i=0}^{j}\si{i}{D, D'} = \sum_{i=0}^{k}(k-i+1)\si{i}{D, D'} \textrm{.}
\end{align*}
We define $\tk{-1}{D}\coloneqq 0$, and introduce the recursive representation for the triple cumulated $k$-edges.
\begin{lemma}\label{lemma:recursive_triple_cumu}
	Let $D$ be a good drawing of $K_n$, $v\in V$ and $F\in \mathcal{F}(D)$. With respect to the reference face $F$ and for all $k\in\{0,\ldots,\lfloor\frac{n}{2}\rfloor-2\}$, we have
	\begin{align*}
		\tk{k}{D} = \tk{k-1}{D-v}+ \tk{k}{D,v} + \bar{I}_{k}(D,D-v) \textrm{.}
	\end{align*}
\end{lemma}
\begin{proof}
	We remove vertex $v$ from the drawing and lose all edges incident to $v$ and therefore also the contribution $\tk{k}{D, v}$. 
	After deleting $v$ every edge is only part of $(n-3)$ triangles in $D-v$ instead of the $(n-2)$ triangles in $D$. 
	Therefore, the $k$-value of an edge may change if the removed triangle accounted to the $k$-value of the edge.
	However, if the edge is invariant, we have to account for an additional contribution.
	The difference for an invariant $i$-edge in $D$ and $D-v$ is
	\[{k+2-i \choose 2}-{k+1-i \choose 2}=(k+1-i)\textrm{.}\]
	This means the contribution equals the double cumulated $k$-value of the $i$-edge, hence the total contribution of all invariant edges is exactly $\di{k}{D, D'} = \sum_{i=0}^{k}(k+1-i)\si{k}{D, D'}$.
	Summing up the terms for the contribution of $v$, the contribution of the invariant edges and the triple cumulated $(k-1)$-value of the drawing $D-v$ leads to the result.
\end{proof}
Using the triple cumulated value, we only have to ensure that
$\tk{k}{D}\geq 3{k+4 \choose 4}$ for $k=\frac{n-1}{2}-2$ if $n$ is odd, or for each $k\in\{\frac{n}{2}-2,\frac{n}{2}-3\}$ if $n$ is even in order to prove that $cr(D)\geq H(n)$ (Theorem \ref{theorem:triplecumubound}).
Mutzel and Oettershagen \cite{seqshellable} showed that for a seq-shellable drawing $D$ of $K_n$ we have  $\dk{i}{D} \geq 3{{i+3}\choose{3}}$ for all $i\in\{0,\ldots,k\}$ with respect to the reference face $F$. This implies the following corollary.
\begin{cor}\label{corollary:seqshell3cumu} 
	Let $D$ be a good drawing of $K_n$ and seq-shellable for a reference face $F\in \mathcal{F}(D)$, then $\tk{k}{D} \geq 3{k+4 \choose 4}$ for all $0\leq k\leq \lfloor \frac{n}{2} \rfloor -2$ with respect to $F$.
\end{cor}
The following lemma gives a lower bound on double cumulated invariant edges incident to a vertex that touches the reference face.
\begin{lemma}\label{lemma:inv_vwatf_cumu2}
	Let $D$ be a good drawing of $K_n$ with two vertices $v$ and $w$ incident to the reference face $F\in\mathcal{F}(D)$.
	If $v$ is removed, the double cumulated value of invariant $k$-edges incident to $w$ with respect to $F$ is at least ${k+2\choose 2}$ for all  $k\in\{0,\ldots,\lfloor\frac{n}{2}\rfloor-2\}$.
\end{lemma}
We use the following statements for the proof of Lemma~\ref{lemma:inv_vwatf_cumu2}.
\begin{cor}\label{corollary:a0b0_j_edge}\emph{\cite{seqshellable}}
	Let $D$ be a good drawing of $K_n$, $F\in \mathcal{F}(D)$ and $u,v\in V$ with both $u$ and $v$ incident to $F$.
	If and only if $uv$ is a $j$-edge, there are exactly $j$ or $n-2-j$ vertices on the same side of $uv$ as the reference face $F$.
\end{cor}
\begin{lemma}\label{lemma:inv_vwatf}\emph{\cite{seqshellable}}
	Let $D$ be a good drawing of $K_n$, $F\in\mathcal{F}(D)$ and $v,w \in V$ with $v$ and $w$ incident to $F$. If we remove $v$ from $D$, then $w$ is incident to at least $\lfloor \frac{n}{2}\rfloor-1$ invariant edges.
\end{lemma}
\begin{proof}
	We label the edges incident to $w$ counter clockwise with $e_0,\ldots,e_{n-2}$ such that $e_0$ and $e_{n-2}$ are incident to the face $F$, 
	and we label the vertex at the other end of $e_i$ with $u_i$. 
	Furthermore, we orient all edges incident to $w$ as outgoing edges.
	Due to Lemma \ref{lemma:vatf} we know that $w$ has two $i$-edges for $0\leq i \leq\lfloor \frac{n}{2}\rfloor -2$. Edge $e_i$ obtains its $i$-value from the minimum of say $+$ oriented triangles and edge $e_{n-2-i}$ obtains its $i$-value from the minimum $-$ oriented triangles (or vice versa).
	Assume that $vw$ is incident to $F$, i.e. $vw$ is a $0$-edge and all triangles $vwu$ for $u\in V\setminus\{v,w\}$ have the same orientation. 
	Consequently, all $e_i$ or all $e_{n-2-i}$ for $0\leq i \leq \lfloor\frac{n}{2}\rfloor-2$ are invariant.
	In the case that $vw$ is not incident to $F$ and is a $j$-edge, there are $j$ triangles $vwu_h$ with $u_h\in V\setminus\{v,w\}$, $0\leq h\leq j-1$ or $n-1-j \leq h\leq n-2$ 
	and $u_h$ is on the same side of $vw$ as $F$ (Corollary \ref{corollary:a0b0_j_edge}). 
	This means, each triangle $wu_hv$ is part of the majority of orientations for the $k$-value of edge $wu_h$, therefore removing $v$ does not change the $k$-value and 	
	there are $j$ additional invariant edges incident to $w$ if we remove $v$.
\end{proof}
\begin{refproof}\textbf{Proof of Lemma \ref{lemma:inv_vwatf_cumu2}.}
	Let $D$ be a good drawing of $K_n$ with vertices $v$ and $w$ incident to the reference face $F\in\mathcal{F}(D)$.
	From the proof of Lemma \ref{lemma:inv_vwatf} follows that if $v$ is removed, vertex $w$ has at least one invariant $(\leq i)$-edge for $0\leq i\leq k$ with respect to $F$.
	Summing up amounts to at least \[\sum_{i=0}^{k}(k+1-i)={k+2 \choose 2}\]
	for all  $k\in\{0,\ldots,\lfloor\frac{n}{2}\rfloor-2\}$.
\end{refproof}
The following lemma is the gist that allows us to locally change the reference face if we have an odd number of vertices.
\begin{lemma}\label{lemma:invarianteinvariant}
	Let $D$ be a good drawing of $K_n$ and $v \in V$.
	For $n$ odd, the value of the double cumulated invariant edges $\di{\mm}{D,D-v}$ is the same with respect to any face incident to $v$ in $D$ and the superface $f(v)$ in $D-v$.
\end{lemma}
\begin{proof}
	Let $m=\mm$.
	With Lemma \ref{lemma:recursive_triple_cumu} follows that with respect to a face incident to $v$ 
	\begin{align*}
		\di{m}{D,D-v}=&\tk{m}{D}-\tk{m-1}{D-v}-\tk{m}{D, v}\text{.}
	\end{align*}	
	$\tk{m}{D}$ is the same for all faces of $D$, the value $\tk{m-1}{D-v}$ with respect to face $f(v)$ is fixed and for each face incident to $v$ we have $\tk{m}{D,v}=2{ m+3 \choose 3}\text{.}$	Therefore, it follows that also the value of $\di{m}{D,D-v}$ has to be the same for every face incident to $v$.
\end{proof}
\section{\sspsyG}\label{sec:pairseqshell}
Basis for the new class of \ssps{} drawings are pair-sequences.
\begin{definition}[Pair-sequence]
	Let $D$ be a good drawing of $K_n$ and $v\in V$. Furthermore, let
	$P_v=(u_0,\ldots ,u_{\lfloor\frac{n}{2}\rfloor-2})$ be a sequence of distinct vertices $u_i\in V\setminus\{v\}$ for $0\leq i\leq \lfloor\frac{n}{2}\rfloor-2$.
	
	We call $P_v$ pair-sequence of $v$ if for \mbox{$j\in \{1,\ldots, \lfloor\frac{n}{2}\rfloor-3\}$} and 
	$(n-j)$ odd, the vertex $u_j$ in the drawing $D-\{u_0,\ldots,u_{j-1}\}$ is incident to a face $F\in\mathcal{F}(D-\{u_0,\ldots,u_{j-1}\})$, where $F$ is also incident to $v$, and in the drawing $D-\{u_0,\ldots,u_j\}$ vertex $u_{j+1}$ is incident to face $f(u_j)$, and vertex $u_0$ is incident to $F\in\mathcal{F}(D)$, where $F$ is also incident \mbox{to $v$.}
\end{definition}
%
For example, in Figure \ref{fig:k11_sss} vertex $v$ in the drawing of $K_{11}$ has the pair-sequence $(u_0,u_1,u_2,u_3)$.
The pair-sequence of vertex $v$ ensures that if we remove $v$ from $D$, there are enough double cumulated invariant $k$-edges.
Therefore, we are able to guarantee a lower bound on $\tk{\mm}{D}$ using Lemma~\ref{lemma:recursive_triple_cumu}.
\begin{lemma}\label{lemma:main2_cumuplus_new}
	Let $D$ be a good drawing of $K_n$, \mbox{$v\in V$} and $(u_0,\ldots ,u_{\lfloor\frac{n}{2}\rfloor-2})$ a pair-sequence of $v$, then $\di{\mm}{D,D-v} \geq {\lfloor\frac{n}{2}\rfloor+1 \choose 3}$. 
\end{lemma}
\begin{proof}
	Without loss of generality let $n$ be odd and let $m=\frac{n-1}{2}-2$ 
	(for $n$ even we can proceed similarly and start with $m=\frac{n}{2}-2$).
	Lemma \ref{lemma:inv_vwatf_cumu2} states that the double cumulated value of invariant edges incident to $u_0$ equals 
	${k+2\choose 2}$ for $0\leq k\leq m$ with respect to a face $F$ incident to $v$ and $u_0$, and the removal of $v$ from $D$.
	Likewise, the double cumulated value of invariant edges incident to $u_1$ is at least ${k+2\choose 2}$ for $0\leq k\leq m-1$ if we remove $v$ from $D-u_0$ with respect to $F$.
	The edge $u_0u_1$ may be invariant or non-invariant in $D$ with respect to removing $v$.
	Now consider the drawing $D-\{u_0,u_1\}$ with $n-2$ vertices and $\frac{n-3}{2}-2=\frac{n-1}{2}-3=m-1$.
	Because $n-2$ is odd, we know that for all faces incident to $v$ the value of $\di{m-1}{D-\{u_0,u_1\},D-\{v,u_0,u_1\}}$ is the same (Lemma \ref{lemma:invarianteinvariant}).
	We may select a new reference face $F'$, such that $v$ and $u_3$ are incident to $F'$, and we can argue again, using Lemma \ref{lemma:inv_vwatf_cumu2}, that removing $v$ leads to at least ${k+2\choose 2}$ for $0\leq k\leq m-2$ double cumulated value of invariant edges incident to $u_2$, since $u_2$ is incident to $F'$. The double cumulated value of invariant edges incident to $u_3$ is at least ${k+2\choose 2}$ for $0\leq k\leq m-3$ with respect to $F'$ if we remove $v$ from $D-\{u_0,u_1,u_2\}$. 
	Again, the edge $u_2u_3$ may be invariant or non-invariant in $D-\{u_0,u_1\}$ with respect to removing $v$.
		
	In general, we are able to change the reference face incident to $v$ if a subdrawing $K_r$ of $K_n$ with $0< r \leq n$ has an odd number of vertices because the number of double cumulated invariant $(\lfloor\frac{r}{2}\rfloor-2)$-edges does not change (see Lemma \ref{lemma:invarianteinvariant}).
	Furthermore, since vertex $u_i$ for $0\leq i \leq \lfloor\frac{n}{2}\rfloor-2$ is incident to the (current) reference face,
	$u_i$ contributes at least ${m-i+2 \choose 2}$ to the value of the double cumulated invariant $m$-value with respect to removing $v$ from $D$. Thus,
	$\di{m}{D,D-v}\geq \sum_{i=1}^{m+2}{i \choose 2}={m+3 \choose 3}$.
\end{proof}
\begin{figure}
	\centering
	\includegraphics[width=0.6\linewidth]{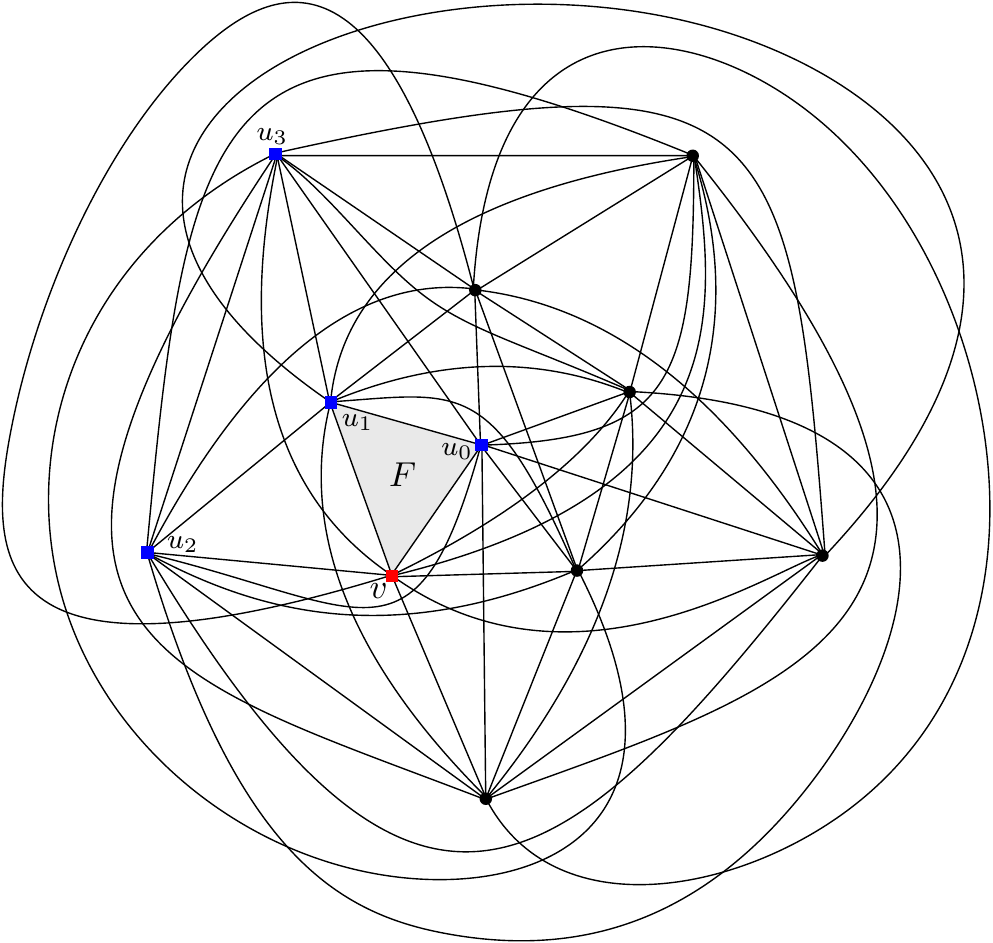}
	\caption{Single-pair-seq-shellable drawing of $K_{11}$. The reference face is $F$, vertex $v$ has the pair-sequence $(u_0,u_1,u_2,u_3)$.}
	\label{fig:k11_sss}
	\vspace{5mm}	
	\includegraphics[width=0.6\linewidth]{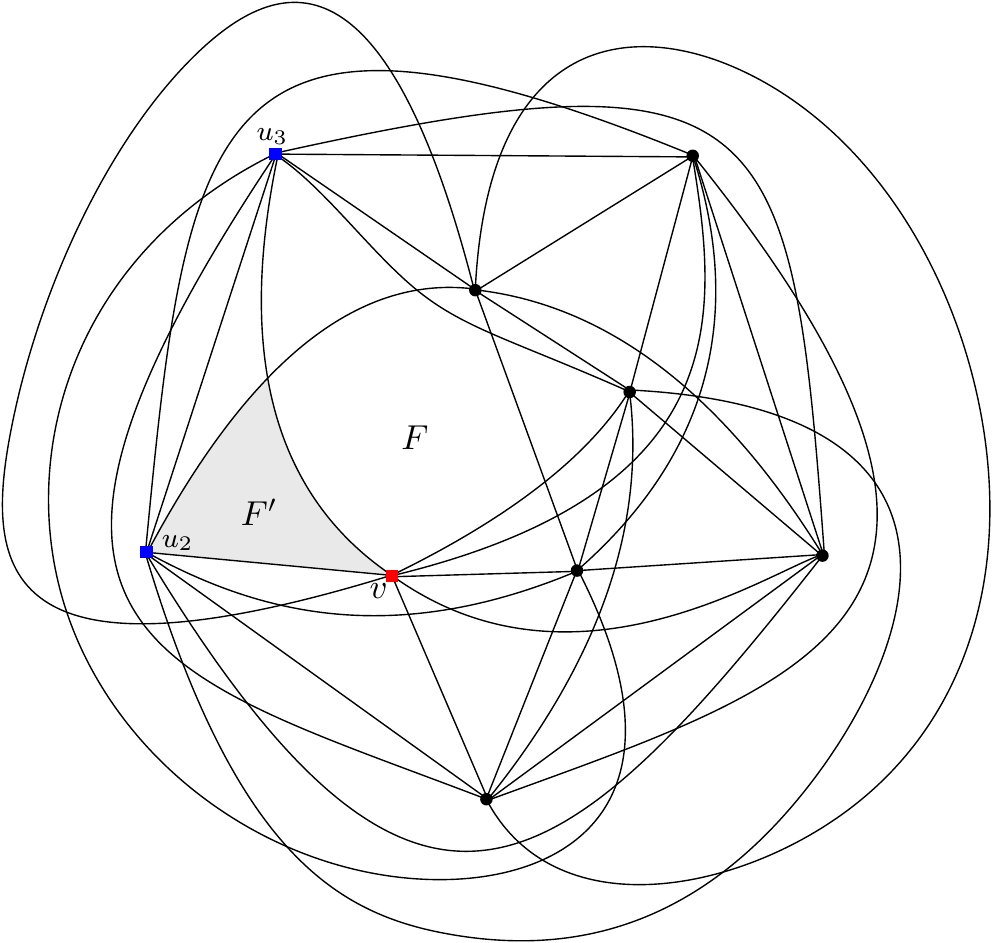}
	\caption{Subdrawing $D-\{u_0,u_1\}$ of the drawing shown in Figure \ref{fig:k11_sss}. The reference face is now $F'$, which is incident to $v$ and $u_2$.}
	\label{fig:k11_sss_2}	
\end{figure}
In Figure \ref{fig:k11_sss}, both vertices $u_0$ and $u_1$ are incident to the initial reference face $F$.
Figure \ref{fig:k11_sss_2} shows the drawing after removing the first pair (i.e. $u_0$ and $u_1$). The face $F$ is not incident to any vertex except $v$.
Changing the reference face to $F'$ allows to proceed with $u_2$ and $u_3$.

Notice, that in a drawing $D$ of $K_n$ with $n$ odd, only the value of $\di{\mm}{D,D-v}$ is invariant with respect to changing the reference face.
The values $\di{k}{D,D-v}$ for $k\in \{0,\ldots,\mmm\}$ may change when selecting a different reference face. 
\begin{lemma}\label{lemma:abi_generalized_01}
	Let $D$ be a good drawing of $K_n$ with $n$ odd and $v\in V$. 
	If $v$ has a pair-sequence and for the subdrawing $D-v$ we have $\tk{\mmm}{D-v}\geq 3{\lfloor\frac{n}{2}\rfloor+1 \choose 4}$ with respect to $f(v)$, then $cr(D)\geq H(n)$.
\end{lemma}
\begin{proof}
	We have $\tk{\mm}{D,v}\geq 2{\lfloor\frac{n}{2}\rfloor+1 \choose 3}$ for any face that is incident to $v$ in $D$, and because $v$ has a pair-sequence and due to Lemma \ref{lemma:main2_cumuplus_new}, it follows that $\di{\mm}{D,D-v}\geq {\lfloor\frac{n}{2}\rfloor+1 \choose 3}$. 
	Using Lemma \ref{lemma:recursive_triple_cumu}, it follows for every face incident to $v$ 
	$\tk{\mm}{D} \geq 3{\lfloor\frac{n}{2}\rfloor+2 \choose 4} \text{.}$
	Since $n$ is odd, the result follows with Corollary \ref{corollary:triple_cumu_lower_bound}.
\end{proof}
Next, we define \sspsy. 
\begin{definition}\label{def:almostseqshellability} 
	Let $D$ be a good drawing of $K_n$ with $n$ odd.
	If there exists a vertex $v\in V$ that has a pair-sequence and the subdrawing $D-v$ is seq-shellable for $f(v)$, then we call $D$ \ssps. 
\end{definition}
Using Lemma \ref{lemma:abi_generalized_01}, we prove the Harary-Hill Conjecture for \ssps{} drawings.
By definition the subdrawing $D-v$ is seq-shellable, hence $\tk{\mmm}{D-v}\geq 3{\lfloor\frac{n}{2}\rfloor+1 \choose 4}$ for $f(v)$ (see Corollary \ref{corollary:seqshell3cumu}). 
Consequently, Theorem \ref{theorem:main_hhc_abi} follows.
\begin{theorem}\label{theorem:main_hhc_abi}
	If $D$ is a \ssps{} drawing of $K_n$, then $cr(D)\geq H(n)$. 
\end{theorem}
The drawing $D$ in Figure \ref{fig:k11_sss} is \ssps{} but not seq-shellable.
It is impossible to find 
a vertex sequence and corresponding simple sequences to apply the definition of seq-shellability.
However, the subdrawing $D-v$ is seq-shellable for face $f(v)$ (see Fig.~\ref{fig:k11_seqshell}) and $v$ has a pair-sequence.
Consequently, $D$ is \ssps.

We are not aware of a crossing optimal \ssps{} drawing that is not seq-shellable.
Every $(\lfloor\frac{n}{2}\rfloor-1)$-seq-shellable drawing $D$ with $n$ odd is also \ssps:
By definition $D$ has a vertex sequence $a_0,\ldots,a_{\lfloor\frac{n}{2}\rfloor-1}$, and each $a_i$ has a simple sequence $S_i$ with $i\in\{0,\ldots,\lfloor\frac{n}{2}\rfloor-1\}$. 
The first $\lfloor\frac{n}{2}\rfloor-2$ vertices of $S_0$ are a pair-sequence for $a_0$.
Moreover, the drawing $D-a_0$ is $(\lfloor\frac{n}{2}\rfloor-2)$-seq-shellable with the vertex sequence 
$a_1,\ldots,a_{\lfloor\frac{n}{2}\rfloor-1}$ and its corresponding simple sequences.
However, there exist $(\lfloor\frac{n}{2}\rfloor-2)$-seq-shellable drawings that are not \ssps{} (see Fig.~\ref{fig:k9_bis_not_almost_bishell}).
Thus, \sspsy{} is a new distinct class that intersects but does not contain the class of seq-shellable drawings.
\begin{figure}[htp!]
	\centering
	\includegraphics[width=0.6\linewidth]{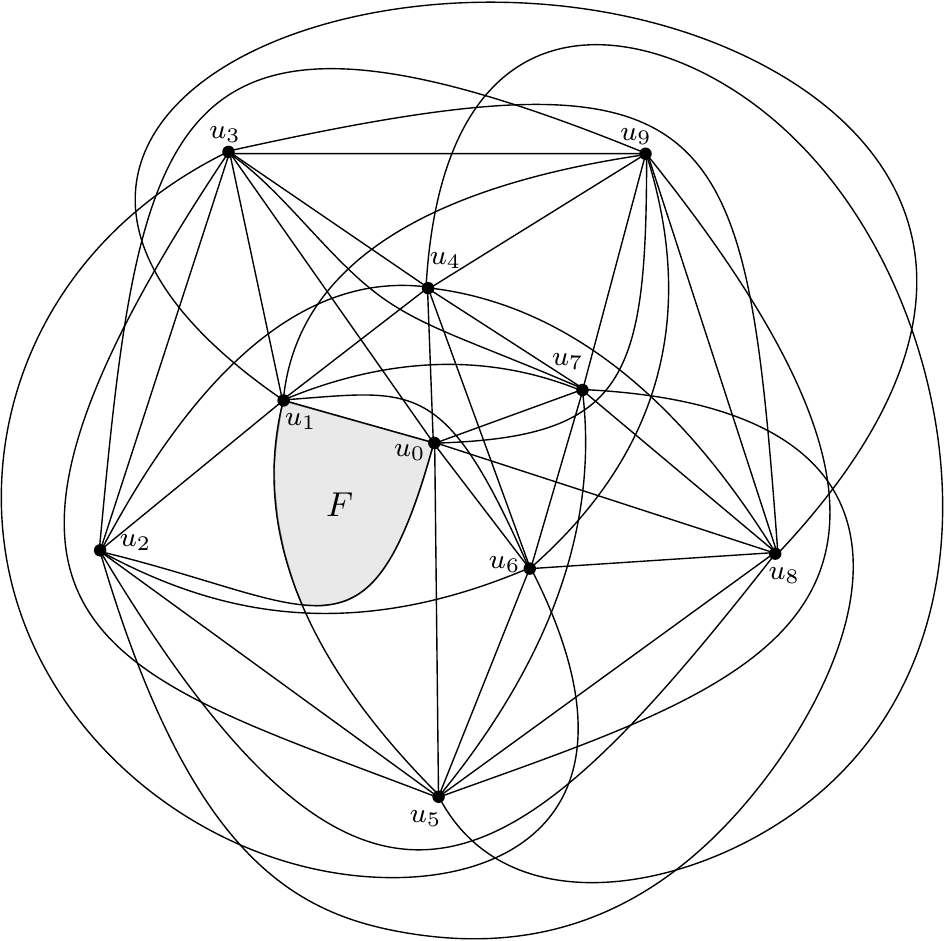}
	\caption{
		Subdrawing $D-v$ of the drawing $D$ of $K_{11}$ shown in Figure \ref{fig:k11_sss}. $D-v$ is seq-shellable for face $F$, vertex sequence $(u_0,u_6,u_5,u_8)$ and the simple sequences $S_0=(u_1,u_2,u_3,u_4)$, $S_1=(u_1,u_2,u_3)$, $S_2=(u_1,u_2)$ and $S_3=(u_1)$.
	}
	\label{fig:k11_seqshell}
	\vspace{5mm}
	\centering
	\includegraphics[width=0.6\linewidth]{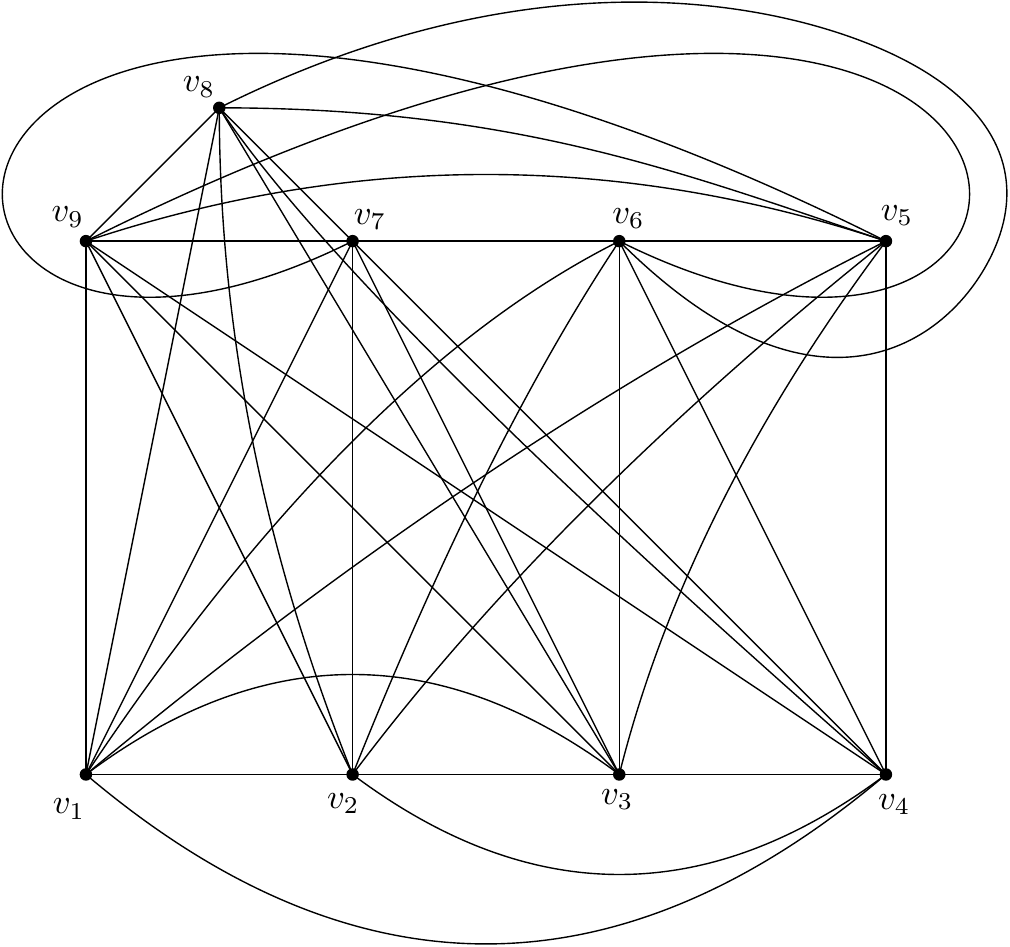}
	\caption{
		Drawing of $K_9$ that is not \ssps{} but $(\lfloor\frac{n}{2}\rfloor-2)$-seq-shellable with respect to the face that is incident to the vertices $v_1,v_2,v_4$. For every vertex $v\in\{v_1,\ldots,v_9\}$ the subdrawing $D-v$ is not seq-shellable with respect to $f(v)$, and therefore $D$ is not \ssps. }
	\label{fig:k9_bis_not_almost_bishell}
\end{figure}	
\section{$k$-Deviations} \label{sec:kdev}
In the following, we introduce $k$-deviations which we use to represent the difference between (cumulated) $k$-edges and optimal values; 
$k$-deviations allow us to formulate conditions under which we are able to change the reference face even more freely. 

Note that if for a drawing $D$ of $K_n$ $\sk{k}{D}=3(k+1)$ for all $0\leq k\leq \lfloor\frac{n}{2}\rfloor-2$, then $cr(D)= H(n)$.
We define $k$-deviations as the difference between this value and the number of $k$-edges in a drawing.
\begin{definition}
	Let $D$ be a good drawing of $K_n$, $F\in \mathcal{F}(D)$ and $E_k(D)$ the number of $k$-edges for $0\leq k \leq \lfloor\frac{n}{2}\rfloor-2$ with respect to $F$. 
	We denote $\sd{k}{D} \coloneqq \sk{k}{D}-3(k+1)$ the $k$-deviation of the drawing $D$ for $0\leq k \leq \lfloor\frac{n}{2}\rfloor-2$ with respect to $F$. 
	Moreover, we define the cumulated versions of the $k$-deviation for $F$ as
	\begin{align*}
	&\dd{k}{D} \coloneqq \displaystyle\sum_{i=0}^{k}\sum_{j=0}^{i}\sd{j}{D} = \sum_{i=0}^{k}(k+1-i)\sd{i}{D}\text{~and}\\ 
	&\td{k}{D} \coloneqq \displaystyle\sum_{i=0}^{k}\dd{i}{D}=\sum_{i=0}^{k}{k+2-i\choose 2}\sd{i}{D}\text{.}
	\end{align*}
	Finally, we define the deviation of the crossing number of $D$ from the Harary-Hill optimal number of crossings as	$\Delta_{cr}(D) \coloneqq cr(D) - H(n)\text{.}$
\end{definition}
We can express $k$-deviations in the following ways.
\begin{lemma}\label{lemma:kdev_01}
	Let $D$ be a good drawing of $K_n$. For a reference face $F\in\mathcal{F}(D)$ and $0\leq k\leq \lfloor\frac{n}{2}\rfloor-2$, we have 
	$\td{k}{D} = \td{k-1}{D} + \dd{k}{D}\textrm{.}$
\end{lemma}
\begin{proof}
	We have for $0\leq k\leq \lfloor\frac{n}{2}\rfloor-2$
	\begin{align*}
	\dk{k}{D} &= 
	\dd{k}{D}+3{k+3 \choose 3} \text{~~and~~}
	\tk{k}{D} = 
	\td{k}{D}+3{k+4 \choose 4}\text{.}
	\end{align*}
	Hence $\tk{k}{D} = \tk{k-1}{D} + \dk{k}{D}$, and 
	it follows
	$\td{k}{D} = \td{k-1}{D} + \dd{k}{D}$
	for all $k\in\{0,\ldots, \mm\}$ (we define $\td{-1}{D}\coloneqq0$).
\end{proof}
\begin{cor}
	\label{corollary:kdev_delta_cr_01}
	Let $D$ be a good drawing of $K_n$. For $n$ odd we have
	$\hd{D} = 2\td{\frac{n-1}{2}-2}{D}$, and for a reference face $F\in  \mathcal{F}(D)$ and $n$ even
	$\hd{D} = \td{\frac{n}{2}-2}{D}+\td{\frac{n}{2}-3}{D}. $
\end{cor}
\begin{proof}
	This result is a direct consequence of the definition of $\Delta_{cr}(D)$, the cumulated $k$-deviations and Theorem \ref{theorem:triplecumubound}.
	For $n$ odd we have
	\begin{align*}
	\hd{D} &= cr(D) - H(n) 
	\\&=2\Bigg(\td{\mm}{D}+3{\lfloor\frac{n}{2}\rfloor+2 \choose 4} \Bigg) - \frac{1}{8}n(n-1)(n-3) -H(n)
	\\&=2\cdot\td{\mm}{D}\text{,}
	\end{align*}
	and for $n$ even
	\begin{align*}
	\hd{D} &= cr(D) - H(n)
	\\&=\Bigg(\td{\mm}{D}+3{\lfloor\frac{n}{2}\rfloor+2 \choose 4} \Bigg)+\Bigg(\td{\mmm}{D}+3{\lfloor\frac{n}{2}\rfloor+1 \choose 4} \Bigg) 
	\\&\hspace{10mm}-\frac{1}{8}n(n-1)(n-2) - H(n)
	\\&=\td{\mm}{D}+\td{\mmm}{D}\text{.}
	\end{align*}
\end{proof}
Notice, that Corollary \ref{corollary:kdev_delta_cr_01} implies Kleitman's parity theorem for complete graphs \cite{kleitman1976note}.
The following lemma gives a lower bound on $\td{\mmm}{D}$.
\begin{lemma}\label{lemma:dbgzero}
	Let $D$ be a good drawing of $K_n$ with $cr(D)\geq H(n)$. For each $F\in \mathcal{F}(D)$ with $\td{\mm}{D} \geq \dd{\mm}{D}$, is 
	$\td{\mmm}{D}\geq 0$. 
\end{lemma}
\begin{proof}
	We consider the cases for $n$ even and $n$ odd separately:
	Let $D$ be a good drawing of $K_n$ with $n$ even and $cr(D)\geq H(n)$. 
	Let $F\in \mathcal{F}(D)$ be a reference face with $\td{\mm}{D} \geq \dd{\mm}{D}$
	and assume $\td{\mmm}{D} < 0$. 
	With Corollary \ref{corollary:kdev_delta_cr_01}, we have
	\[\hd{D}=\td{\mm}{D}+\td{\mmm}{D} \geq 0\] and it follows 
	\[\td{\mm}{D}\geq |\td{\mmm}{D}|>0\textrm{.}\]
	From Lemma \ref{lemma:kdev_01} follows that 
	\[\td{\mm}{D} = \td{\mmm}{D} + \dd{\mm}{D}\text{.}\]	
	Therefore, 
	\begin{align*}
	\dd{\mm}{D}&=\td{\mm}{D} - \td{\mmm}{D}> \td{\mm}{D}\textrm{,}
	\end{align*}
	a contradiction to $\td{\mm}{D} \geq \dd{\mm}{D}$.
	Now, let $D$ be a good drawing of $K_n$ with $n$ odd and $cr(D)\geq H(n)$. 
	Let $F\in \mathcal{F}(D)$ be a reference face with $\td{\mm}{D} \geq\dd{\mm}{D}$
	and assume $\td{\mmm}{D} < 0$. With Corollary \ref{corollary:kdev_delta_cr_01} and due to 
	\[\hd{D}=2\td{\mm}{D}\geq 0\]
	we have 
	$\td{\mm}{D}\geq0 \textrm{~~and~~} \td{\mmm}{D}+\dd{\mm}{D}\geq 0 \textrm{.}$
	From $\td{\mmm}{D} < 0$ and
	\[\td{\mm}{D}=\td{\mmm}{D}+ \dd{\mm}{D}\] follows a contradiction to $\td{\mm}{D} \geq \dd{\mm}{D}$.
\end{proof}
With the following proposition, we are able to select a new reference face for the subdrawing $D-v$.
\begin{prop}\label{theorem:conjtheorem}
	Let $D$ be a good drawing of $K_n$ with $n$ odd and $v\in V$, such that the subdrawing $D-v$ is seq-shellable for any face $F\in \mathcal{F}(D-v)$. 
	If $v$ has a pair-sequence and in subdrawing $D-v$ for $f(v)$ 
	$\td{\frac{n-1}{2}-2}{D-v} \geq \dd{\frac{n-1}{2}-2}{D-v}\textrm{,}$
	then $cr(D)\geq H(n)$.
\end{prop}
\begin{proof}
	Let $D$ be a drawing $K_n$ with $n$ odd and $v\in V$ such that $D-v$ is seq-shellable with respect to a face $F\in\mathcal{F}(D-v)$.
	For the subdrawing $D-v$, which has an even number of $n-1$ vertices, we know that the sum  
	\[\tk{\frac{n-1}{2}-2}{D-v}+\tk{\frac{n-1}{2}-3}{D-v}\] 
	is the same for all faces of $D-v$.
	Due to the seq-shellability and Corollary \ref{corollary:seqshell3cumu}, it follows that $\td{\frac{n-1}{2}-2}{D-v}\geq 0$ and $\td{\frac{n-1}{2}-3}{D-v}\geq 0$ for $F$ and hence for all faces
	\[\td{\frac{n-1}{2}-2}{D-v}+\td{\frac{n-1}{2}-3}{D-v}\geq 0 \text{.}\]
	%
	Because $cr(D-v)\geq H(n-1)$ and $\td{\frac{n-1}{2}-2}{D-v} \geq \dd{\frac{n-1}{2}-2}{D-v}$ 
	for the superface $f(v)$ of $v$, we have with Lemma \ref{lemma:dbgzero} $\td{\frac{n-1}{2}-3}{D-v}\geq 0$ for $f(v)$.
	Since $v$ is incident to a face $F'\subset f(v)$ and $v$ has a pair-sequence, which ensures $\di{\frac{n-1}{2}-2}{D,D-v}\geq {\frac{n-1}{2}+1 \choose 3}$ it follows from Lemma \ref{lemma:recursive_triple_cumu} for face $F'$ that
	$\tk{\frac{n-1}{2}-2}{D} \geq 3{\frac{n-1}{2}+2 \choose 4}$.
	With Corollary \ref{corollary:triple_cumu_lower_bound} follows the result. 
\end{proof}
So far, for all drawings and all faces we inspected, the condition of Lemma \ref{lemma:dbgzero} had been fulfilled.
We conjecture it to be true for all good drawings of $K_n$.
\begin{conj}\label{con:con1}
	Let $D$ be a good drawing of $K_n$. With respect to any face $F\in\mathcal{F}(D)$, we have
	\[\td{\mm}{D} \geq \dd{\mm}{D}\textrm{.}\]
\end{conj}
Under the assumption that Conjecture \ref{con:con1} holds, we are able to prove the Harary-Hill Conjecture for another new class of drawings that comprises the classes of seq-shellable and \ssps{} drawings. Here, we can select a different reference face for each vertex.
\begin{theorem}\label{theorem:conjmain}
	Let $D$ be a good drawing of $K_n$ and $v_1,\ldots,v_n$ a sequence of the vertices, such that every vertex $v_i$ with $i\in \{1,\ldots,n\}$ and $i$ odd has a pair-sequence, and every vertex $v_i$ with $i\in \{1,\ldots,n\}$ and $i$ even has a simple sequence. If Conjecture \ref{con:con1} holds, then $cr(D)\geq H(n)$.
\end{theorem}
The following two lemmas are necessary for our proof of Theorem \ref{theorem:conjmain}.
\begin{lemma}\label{lemma:inv_vwatf_seq}\cite{seqshellable}
	Let $D$ be a good drawing of $K_n$, $F\in\mathcal{F}(D)$ and $v\in V$ with $v$ incident to $F$. 
	If $v$ has a simple sequence $S_v = (u_0,\ldots,u_k)$, then
	\begin{align*}
	\sum_{j=0}^{k}I_{j}(D,D-v)\geq {k+2\choose 2}
	\end{align*}
	with respect to $F$ and for each $k\in \{0,\ldots,\lfloor\frac{n}{2}\rfloor-2\}$.
\end{lemma}

\begin{lemma}\label{lemma:conj_conseq}
	Let $D$ a good drawing of $K_n$ with $n$ even and $v\in V$, such that for subdrawing $D-v$ 
	\[\tk{\lfloor\frac{n-1}{2}\rfloor-2}{D-v}=\tk{\frac{n}{2}-3}{D-v} \geq 3{\frac{n}{2}+1 \choose 4}\text{.}\]
	If $v$ has a simple sequence and Conjecture \ref{con:con1} holds, then $cr(D)\geq H(n)$.
\end{lemma}
\begin{proof}
	Because $D-v$ is a drawing of $K_{n-1}$ with $n-1$ odd, it follows with Corollary \ref{corollary:triple_cumu_lower_bound} that $cr(D-v)\geq H(n-1)$, and that
	$\tk{\frac{n}{2}-3}{D-v} \geq 3{\frac{n}{2}+1 \choose 4}$
	is the same for every face of $D-v$. Lemma \ref{lemma:dbgzero} (together with Conjecture \ref{con:con1}) implies
	\[\tk{\frac{n}{2}-4}{D-v} \geq 3{\frac{n}{2} \choose 4}\text{.}\] 
	Since vertex $v$ has a simple sequence in $D$, Lemma \ref{lemma:inv_vwatf_seq} implies that 
	\begin{align*}
	\di{\frac{n}{2}-2}{D,D-v}\geq {\frac{n}{2}+1 \choose 3} \text{~~and~~}\di{\frac{n}{2}-3}{D,D-v}\geq {\frac{n}{2} \choose 3}\text{.}
	\end{align*} 
	Furthermore, for a face $F\in \mathcal{F}(D)$ that is incident to $v$ is $\tk{k}{D,v}\geq 2{k+3 \choose 3}$ for all $k\in \{0,\ldots, \frac{n}{2}-2 \}$.
	Therefore, in drawing $D$ with respect to a face $F\in \mathcal{F}(D)$ incident to $v$ it follows that
	\begin{align*}
	\tk{\frac{n}{2}-2}{D} &= \tk{\frac{n}{2}-3}{D-v} + \tk{\frac{n}{2}-2}{D,v} +\di{\frac{n}{2}-2}{D,D-v}\\
	&\geq 3{\frac{n}{2}+1 \choose 4}+2{\frac{n}{2}+1 \choose 3}+{\frac{n}{2}+1 \choose 3} 
	\\&= 3{\frac{n}{2}+2 \choose 4} \text{~and}\\
	\tk{\frac{n}{2}-3}{D} &= \tk{\frac{n}{2}-4}{D-v} +\tk{\frac{n}{2}-3}{D,v} 
	+\di{\frac{n}{2}-3}{D-v}\\
	&\geq 3{\frac{n}{2} \choose 4}+2{\frac{n}{2} \choose 3}+{\frac{n}{2} \choose 3}= 3{\frac{n}{2}+1 \choose 4}\text{.}
	\end{align*}
	Finally, with Corollary \ref{corollary:triple_cumu_lower_bound} follows that $cr(D)\geq H(n)$.
\end{proof}
\begin{refproof}\textbf{Proof of Theorem \ref{theorem:conjmain}.}
	We start with the case for $n$ odd.
	We proceed with induction over the number of vertices.
	\\\textbf{Basis:} For $n=3$ there are no crossings, thus $cr(D)= 0=H(3)$, and
	for $n=4$ we have either $cr(D)= 0=H(3)$ or $cr(D)= 1 \geq 0 = H(3)$.
	\\\textbf{Induction step:}
	Let $D$ be a drawing fulfilling the requirements, i.e. $D$ has a sequence $v_1,\ldots,v_n$ of the vertices, such that every vertex $v_i$ with $i$ odd has a pair-sequence, and every vertex $v_i$ with $i$ even has a simple sequence. If we remove vertices $v_n$ and $v_{n-1}$, then the subdrawing $D-\{v_{n-1},v_n\}$ still fulfills the requirements for the sequence $v_1,\ldots,v_{n-2}$, and we assume that 
	\begin{align*}
	\tk{\lfloor\frac{n-2}{2}\rfloor-2}{D-\{v_{n-1},v_n\}} = \tk{\frac{n-1}{2}-3}{D-\{v_{n-1},v_n\}} \geq 3{\frac{n-1}{2}+1 \choose 4}
	\end{align*}
	Using this assumption with Lemma \ref{lemma:conj_conseq}, we have $cr(D-v_n)\geq H(n-1)$. 
	Lemma \ref{lemma:dbgzero} (together with Conjecture \ref{con:con1}) implies that 
	$\tk{\frac{n-1}{2}-3}{D-v_n}\geq 3{\frac{n-1}{2}+1 \choose 4}$
	for all faces of $D-v_n$.
	Because $v_n$ has a pair-sequence, it follows with Lemma \ref{lemma:main2_cumuplus_new} that $\di{\frac{n-1}{2}-2}{D,D-v_n}\geq {\frac{n-1}{2}+1 \choose 3}$, 
	and for any face incident to $v_n$ we have $\tk{\frac{n-1}{2}-2}{D, v_n}\geq 2{\frac{n-1}{2}+1 \choose 3}$.
	Altogether, it follows 
	\begin{align*}
	\tk{\frac{n-1}{2}-2}{D}&=\tk{\frac{n-1}{2}-3}{D-v_n}+\tk{\frac{n-1}{2}-2}{D,v_n} + \di{\frac{n-1}{2}-2}{D,D-v_n}
	\\&\geq 3{\frac{n-1}{2}+1 \choose 4}+2{\frac{n-1}{2}+1 \choose 3}+{\frac{n-1}{2}+1 \choose 3}\\&= 3{\frac{n-1}{2}+2 \choose 4}\text{.}
	\end{align*}
	Therefore, $\tk{\frac{n-1}{2}-2}{D}\geq 3{\frac{n-1}{2}+2 \choose 4}$ for all faces of $D$, and since $n$ is odd, it follows from Corollary \ref{corollary:triple_cumu_lower_bound} that $cr(D)\geq H(n)$. 
	
	For the case $n$ even we remove vertex $v_n$ and have a drawing with an odd number of $n-1$ vertices that fulfills the requirements. Therefore, 
	$\tk{\lfloor\frac{n-1}{2}\rfloor-2}{D-v_n}=\tk{\frac{n}{2}-3}{D-v_n}\geq 3{\frac{n}{2}+1 \choose 4}$ for all faces of drawing $D-v_n$. With Lemma \ref{lemma:conj_conseq} follows $cr(D)\geq H(n)$.
\end{refproof}
\section{Conclusions and Outlook}\label{sec:conclusions}
We introduced \ssps{} drawings of complete graphs and verified the Harary-Hill Conjecture for this new class.
For the first time, we used more than a single globally fixed reference face in order to show lower bounds on the triple cumulated $k$-edges.
\Sspsy{} is only defined for drawings of $K_n$ with $n$ odd so far. 
Extending \sspsy{} to drawings of $K_n$ with an even number of vertices is an open problem. 
Here, it would suffice to show that \mbox{$\td{\mm}{D}+\td{\mmm}{D}\geq 0$} implies \mbox{$\td{\mmm}{D}\geq 0$} in order to generalize our results from \sspsy{} to \emph{\psy}, i.e. a version of seq-shellability with pair-sequences instead of simple sequences.
Moreover, we introduced $k$-deviations to formulate conditions under which we are able to select a new reference face in each subdrawing. 
Proving Conjecture \ref{con:con1} would settle the Harary-Hill Conjecture for a very broad class of drawings, comprising seq- and \sspsy.
Still, there are optimal drawings where each face touches a single vertex only \cite{DBLP:conf/cccg/AbregoAF0V14}, thus no vertex has a simple or pair-sequence.
%
%
%
\small
\bibliographystyle{abbrv}

\end{document}